\documentclass{svproc}

\renewcommand{\bfseries}{\fontseries{b}\selectfont}

\usepackage{graphicx}%
\usepackage{multirow}%
\usepackage{amsmath,amssymb,amsfonts}%
\usepackage{mathrsfs}%
\usepackage[title]{appendix}%
\usepackage{xcolor}%
\usepackage{textcomp}%
\usepackage{manyfoot}%
\usepackage{booktabs}%
\usepackage{algorithm}%
\usepackage{algorithmicx}%
\usepackage{algpseudocode}%
\usepackage{listings}%
\usepackage{url}
\usepackage{comment}
\usepackage{subcaption}
\usepackage{optidef}

\newcommand{\vol}{\mathrm{vol}}
\newcommand{\cut}{\mathrm{cut}}
\newcommand{\cutH}{\mathrm{cut}_\mathcal{H}}

\begin{document}
\mainmatter              
\title{Hypergraph Change Point Detection using Adapted Cardinality-Based Gadgets: \\
Applications in Dynamic Legal Structures}
\titlerunning{Hypergraph Change Point Detection}  
%
\author{Hiroki Matsumoto\inst{1} \and Takahiro Yoshida\inst{2} \and Ryoma Kondo\inst{1,2} \and \\ 
Ryohei Hisano\inst{1,2}}
\authorrunning{Hiroki Matsumoto et al.} 

\institute{Graduate School of Information Science and Technology, the University of Tokyo, 7-3-1 Hongo, Bunkyo-ku, Tokyo, 113-8656, Japan\\
\email{hisanor@g.ecc.u-tokyo.ac.jp}
\and
The Canon Institute for Global Studies,
11th Floor, ShinMarunouchi Building, 5-1 Marunouchi 1-chome, Chiyoda-ku, Tokyo 100-6511, Japan}

\maketitle              

\begin{abstract}

Hypergraphs provide a robust framework for modeling complex systems with higher-order interactions. However, analyzing them in dynamic settings presents significant computational challenges. To address this, we introduce a novel method that adapts the cardinality-based gadget to convert hypergraphs into strongly connected weighted directed graphs, complemented by a symmetrized combinatorial Laplacian. We demonstrate that the harmonic mean of the conductance and edge expansion of the original hypergraph can be upper-bounded by the conductance of the transformed directed graph, effectively preserving crucial cut information. Additionally, we analyze how the resulting Laplacian relates to that derived from the star expansion. Our approach was validated through change point detection experiments on both synthetic and real datasets, showing superior performance over clique and star expansions in maintaining spectral information in dynamic settings. Finally, we applied our method to analyze a dynamic legal hypergraph constructed from extensive United States court opinion data.

\end{abstract}

\section{Introduction}

Hypergraphs offer a natural framework for analyzing complex relationships where the interactions of interest extend beyond only pairs. This makes them particularly well-suited for capturing the higher-order interactions of systems where the relationships are not merely one-to-one but rather involve sets of elements working together. For example, a hypergraph can be used to model collaborations among academics~\cite{Benson2018}, cosponsorship of congress bills~\cite{Fowler2006}, COVID-19 viral protein interactions~\cite{Gopalakrishnan2022}, and the functional connectivity of the human brain~\cite{Ma2021}. By accommodating these higher-order interactions, hypergraphs offer a richer and more flexible representation of complex systems, enabling deeper insights and more accurate  analyses.

Dynamic extensions of hypergraph modeling are also gradually gaining popularity. In \cite{Benson2018}, the authors explored  the temporal evolution of higher-order interactions within complex systems and identified  consistent patterns across various types of systems. They suggested higher-order link prediction as a benchmark to enhance the understanding and modeling of these structures. In \cite{Ito2023}, the authors introduced a novel hypergraph embedding method inspired by modularity maximization. This method visualizes structural changes by positioning hypernodes on concentric spheres. It achieved superior  performance in terms of spatial efficiency and the detection of structural changes. In \cite{Coupette2024}, the authors presented  hypergraphs as a powerful tool for analyzing higher-order interactions in legal networks. Through case studies spanning 70 years, they demonstrated the method's potential in legal citation and collaboration networks, which underscored its significance for advancing legal network analysis. However, modeling and analyzing large, dynamic hypergraphs remains a significant challenge because of their immense computational complexity.

To address the computational complexity in analyzing dynamic networks, such as change point detection, an effective strategy is to focus on the spectral properties of the Laplacian matrix. In \cite{Huang2020}, the authors introduced the Laplacian anomaly detection (LAD) method, which uses these spectral properties to create low-dimensional embeddings of graph snapshots. This approach effectively addresses the challenges associated with comparing graphs over time and capturing their temporal dependencies. In a subsequent study \cite{Huang2021}, the authors further refined this method using efficient approximation techniques related to the network density of states. This enhancement integrates the full spectrum of Laplacian analysis, thereby achieving performance that is comparable to state-of-the-art methods while significantly increasing processing speed. 


However, defining the Laplacian for hypergraphs is a challenging task. In \cite{Agarwal2006}, the authors demonstrated that many existing hypergraph constructions correspond to either the clique or star expansion with adjusted weighting functions. A theoretically sound definition that preserves conductance, in the spirit of \cite{Chung1997,Chung2005}, exists in the literature \cite{Takai2020}; however, it leads to a nonlinear operator, which complicates dynamic hypergraph analysis. Recent works \cite{Liu2021,Veldt2022} on directed graph reduction techniques offer a promising approach. In \cite{Liu2021}, the authors showed that by replacing each hyperedge with a set of cardinality-based gadgets (CB-gadgets), a hypergraph can be transformed into a directed graph  that preserves conductance. This implies that minimizing conductance in the transformed directed graph also minimizes conductance in the original hypergraph. However, the resulting directed network, as shown in \cite{Liu2021}, includes auxiliary nodes where the degree is assumed to be zero, which leads to an artificial network that is not suitable for direct spectral analysis.

In this paper, we adapt the CB-gadget to transform a hypergraph into a strongly connected weighted directed graph. Then we demonstrate that the harmonic mean of the conductance and edge expansion of the original hypergraph can be upper bounded by the conductance of the transformed weighted directed graph. This indicates that the transformed graph approximately preserves the original cut information. By ensuring that the transformed directed graph is strongly connected whenever the original hypergraph is connected, we further refine the weighted directed graph using the combinatorial Laplacian from \cite{Chung2005}. 
Furthermore, we examine the relationship between the resulting Laplacian and the one obtained through the star expansion using a similar analysis to that in \cite{Agarwal2006}.


To demonstrate that the spectrum derived from our resulting Laplacian better preserves the original hypergraph information, we validated our approach through change point detection using the Laplacian spectrum \cite{Huang2020}. We tested our method on both synthetic and real datasets. For the synthetic dataset, we simulated four ground-truth change points and sampled each segment's hypergraph using the method from \cite{Ruggeri2023}. Our results indicated that the proposed expansion method outperformed star and clique expansions in detecting these ground-truth change points. For the real dataset, we used United States court opinion data published by the Free Law Project \cite{FreeLawProject}. We demonstrated the effectiveness of our model by showing that it detected critical comments during significant shifts in the legal hypergraph.



\section{Methods}

\subsection{Hypergraph Conductance and Graph Reduction}

We consider the cut problem in hypergraphs. In the case of standard graphs, where each edge $e$ has a weight $w_e$,  and $e$ always connects exactly two nodes, the cut score of $e$ can simply be represented by $w_e$ when $e$ is part of the cut (thus naturally defining a cut function that assigns a cut score to a given subset $S \subset V$). However, in the case of hypergraphs, there is greater flexibility in how the cut score for subset $S$ can be assigned. To address the $s$-$t$ cut problem with this flexibility, in \cite{Veldt2022}, the authors introduced the concept of an edge-splitting function, which enables the formulation of a generalized $s$-$t$ cut problem in hypergraphs.


With each hyperedge, we associate a splitting function $f_e$ that we use to assess an appropriate penalty for splitting the hyperedge between two clusters. With an identified $f_e$, the cut value of any given set $S\subset V$ can be written as
\begin{equation}
    \cutH(S) = \sum_{e\in E}f_e(e\cap S) = \sum_{e\in \partial S}f_e(e\cap S) .     
\end{equation}

Under this definition, the hypergraph minimum $s$-$t$ cut problem is defined as follows:

\begin{mini}
    {S\subset V}          
    {\cutH(S, V\backslash S)} 
    {\label{hypergraph_st_cut}}  
    {} 
    \addConstraint{s \in S, t \in V\backslash S}. 
\end{mini}

Then we denote the (hypergraph) conductance by $\phi_{\mathcal{H}}(S)$ and the (hypergraph) edge expansion by $\psi_{\mathcal{H}}(S)$: 

\begin{equation}
    \phi_{\mathcal{H}}(S) = \frac{\cutH(S)}{\min\{\vol_{\mathcal{H}}(S), \vol_{\mathcal{H}}(V\backslash S)\}}, \; \psi_{\mathcal{H}}(S) = \frac{\cutH(S)}{\min\{|S|, |V\backslash S|\}},
\end{equation}


\noindent where $d_i = \sum_{e : i \in e} f_e(\{i\})$ and $\vol_{\mathcal{H}}(S) = \sum_{i \in S} d_i$. We define the conductance and edge expansion of $\mathcal{H}$ in a manner similar to that for $G$. It is well established that optimizing conductance, which involves solving the cut problem to minimize the conductance of a hypergraph, is inherently challenging. Although a theoretically sound definition that preserves conductance exists \cite{Takai2020}, it results in a nonlinear operator that is difficult to handle in a dynamic setting.

Another approach involves using graph reduction techniques to minimize hypergraph conductance \cite{Liu2021,Veldt2022}. In \cite{Liu2021}, the authors demonstrate that by replacing each hyperedge with a set of CB-gadgets, a hypergraph can be transformed into a directed graph while preserving conductance. This transformation implies that minimizing conductance in the resulting directed graph also minimizes conductance in the original hypergraph. However, to fully preserve hypergraph conductance, it is necessary to assume that the auxiliary nodes have zero degree (see Theorem 3.3 of \cite{Liu2021}), which complicates further spectral analysis of the resulting directed network, as it can no longer be treated as a strongly connected transition probability matrix of a directed network.

Therefore, based on the CB-gadget, we propose an adapted CB-gadget described as follows:

\begin{definition}[Adapted CB-gadget]
 Given a hypergraph $\mathcal{H} = (V, E)$, we define a set of auxiliary nodes $V_a$ and $V_b$, along with a set of directed edges $\hat{E}$, as follows\footnote{We assign two unique auxiliary nodes to each hyperedge.}:
\begin{enumerate}
\item For each $e \in E$, add two auxiliary nodes $a \in V_a$ and $b \in V_b$.
\item For each $v \in e$, add edges $(v, a)$ and $(b, v)$ with $w_{va}=w_{bv}=1$ to $\hat{E}$.
\item Add an edge $(a,a)$ and $(a, b)$ with $w_{aa} = 1-w_e,\, w_{ab}=w_e, $ to $\hat{E}$,
\end{enumerate}
where $w_e$ is normalized to be in $(0,1]$. 
\end{definition}
The resulting graph $G$ has the node set $\hat{V} = V \cup V_a \cup V_b$. This adaptation does not change the cut score and $w_e\leq 1$;  thus, $f_e$ becomes an all-or-nothing function, as shown in \cite{Veldt2022}:

\begin{equation}
    f_e(S) = \min\{|S|, |e\backslash S|, w_e\} = w_e.
\end{equation}

\noindent Our adaptation enables the consideration of probability transitions for a random walk that reflect the weights in the resulting graph.

In the following, we prove that the weighted harmonic mean of hypergraph conductance and edge expansion can be upper bounded   by twice the conductance of the directed graph obtained from the adapted CB-gadget. Before stating the main theorem, we first define $\mu_{\mathcal{H}}$ as follows:

\begin{equation}
    \mu_{\mathcal{H}}(S) = 2\frac{\cutH(S)}{\vol_{\mathcal{H}}(S) + \beta |S|},
\end{equation}

\noindent where $\beta = (\epsilon + 1)\nu$. We define $\epsilon = \max_{e\in E} |e|$ and $\nu = \max_{v\in S} |\{e\in E \mid v\in e\}|$. The following proposition shows that $\mu_{\mathcal{H}}(S)$ equals the harmonic mean of $\phi_{\mathcal{H}}(S)$ and $\frac{1}{\beta}\psi_{\mathcal{H}}(S)$ under reasonable conditions, which implies that it contains critical information about the hypergraph cut.

\begin{proposition}
If $\vol_{\mathcal{H}}(S)\leq \vol_{\mathcal{H}}(V\backslash S)$ and $|S| \leq |V\backslash S|$, then $\mu_{\mathcal{H}}(S)$ equals the harmonic mean of $\phi_{\mathcal{H}}(S)$ and $\frac{1}{\beta}\psi_{\mathcal{H}}(S)$.
\end{proposition}
\begin{proof}
    \begin{align*}
        \frac{1}{\phi_{\mathcal{H}}(S)} + \frac{\beta}{\psi_{\mathcal{H}}(S)} &= \frac{\min\{\vol_{\mathcal{H}}(S), \vol_{\mathcal{H}}(V\backslash S)\}}{\cutH(S)} + \beta  \frac{\min\{|S|, |V\backslash S| \}  }{\cutH(S)}\\
        &= \frac{\vol_{\mathcal{H}}(S) + \beta S}{\cutH(S)} = \frac{2}{\mu_{\mathcal{H}}(S)}
    \end{align*}
    \begin{equation*}
        \Rightarrow \mu_{\mathcal{H}}(S) = 2\left( \frac{\psi_{\mathcal{H}}(S)+\beta\phi_{\mathcal{H}}(S)}{\psi_{\mathcal{H}}(S)\phi_{\mathcal{H}}(S)} \right)^{-1} = 2\frac{\phi_{\mathcal{H}}(S)\cdot \psi_{\mathcal{H}}(S)/\beta}{\phi_{\mathcal{H}}(S)+\psi_{\mathcal{H}}(S)/\beta}.
    \end{equation*}
\end{proof}


Using this quantity, we prove our main result. Let $\mathcal{T}$ be the set of $T\subset \hat{V}$ such that $\forall u\in T\cap(V_a\cup V_b)$ has an $u\in V$ and $u$ is $u$'s out-neighbor or in-neighbor. Formally, we define

\begin{equation}
    \mathcal{T} := \{ T\subset \hat{V} \mid \forall u\in T\cap(V_a\cup V_b),\, \exists v\in T\cap V,\, (u, v)\in \hat{E} \lor (v, u)\in \hat{E} \}.
\end{equation}

Now we show that $\mu_{\mathcal{H}}(S)$ can be upper bounded by twice the conductance of $G$.


\begin{theorem}
    $\forall T\subset \mathcal{T} \subset V$ and $S=T\cap V$,
    \begin{equation}
        2\phi_G(T)\geq \mu_{\mathcal{H}}(S).
    \end{equation}
\end{theorem}

\begin{proof}
    Without loss of generality, we assume that $\vol_{\mathcal{H}}(T)\leq \vol_{\mathcal{H}}(\hat{V}\backslash T)$. 
    
    Then    \begin{equation}
        \phi_G(T) = \frac{\cut_G(T)}{\min\{ \vol_G(T), \vol_G(\hat{V}\backslash T) \}} = \frac{\cut_G(T)}{\vol_G(T)}
    \end{equation}
    where $\vol_G(T)$ denotes the sum of the out-degrees of each $v\in T$, and 
    \begin{align}
        &\cutH(S) = \min_{U\subset \hat{V}:S=U\cap V}\cut_G(U) \quad (\because \cite{Liu2021})\\
        &\Rightarrow \cut_G(T)\geq \cutH(S).
    \end{align}

    Recall that $\epsilon = \max_{e\in E} |e|$ and $\nu = \max_{v\in S} |\{e\in E \mid v\in e\}|$. Then we have
    \begin{align}
        \vol_G(T) &= \sum_{v\in T}d_v = \sum_{v\in (T\cap V)}d_v + \sum_{v\in (T\cap V_a)}d_v + \sum_{v\in (T\cap V_b)}d_v\\
        &\leq \sum_{v\in S}d_v + \sum_{v\in (T\cap V_a)}1 + \sum_{v\in (T\cap V_b)}\epsilon = \vol_{\mathcal{H}}(S) + |T\cap V_a| + \epsilon|T\cap V_b|\\
        &\leq \vol_{\mathcal{H}}(S) + (\epsilon+1)\nu|S|
    \end{align}

    \begin{equation}
        \therefore 2\phi_G(T) = \frac{2\cut_G(T)}{\vol_G(T)}\geq \frac{2\cutH(T)}{\vol_{\mathcal{H}}(S)+\beta |S|} = \mu_{\mathcal{H}}(S).
    \end{equation}
\end{proof}

Consequently, if a subset $ T \subset \mathcal{T}$ can be identified that achieves sufficiently low conductance on $ G $, then by considering $ S = T \cap V $, we can obtain a vertex set that achieves a correspondingly small $ \mu_{\mathcal{H}} $. This provides a solid justification for optimizing the directed graph reduced by the adapted CB gadget.

The resulting directed graph can be easily shown to be strongly connected if the hypergraph is connected. Since this follows directly from the definition of the CB-gadget, we omit the proof for brevity. This guarantees the existence of a stationary distribution for a random walk on $G$ when $H$ is connected.


\subsection{Comparison Between Our Approach and Star Expansion}

Perhaps the simplest method for performing spectral analysis on directed graphs is Chung's approach \cite{Palmer2021}, which is well-established for preserving cut information. Our adaptation of the CB-gadget enables us to treat each edge weight as a transition probability. We directly apply Chung's combinatorial Laplacian (i.e., $L := \Phi - \frac{\Phi P + P^\top \Phi}{2}$), where $P$ is the transition probability matrix for a random walk on the directed graph, and $\Phi$ is a diagonal matrix containing the unique stationary distribution vector. This transformation is well known for preserving cut information, especially for strongly connected components~\cite{Yoshida2016}.


In the following, we show that under the following situation our proposed Laplacian and that derived from star expansion can fall under the same framework of \cite{Agarwal2006}. Let $D = \mathrm{diag}(d_1,...,d_n)$ and $\Delta = \mathrm{diag}(\delta(1),...,\delta(m))$ with $\delta(e)=|e|$ for $e\in E$. Our resulting Laplacian $L$ derived above and the corresponding transition matrix $P$ can be written as follows.

\begin{equation}
    L = \Phi - \frac{\Phi P+ P^\top\Phi}{2} 
      = \begin{pmatrix}
           \Phi_v & -\frac{\Phi_vD^{-1} H}{2} & -\frac{H \Delta^{-1} \Phi_b}{2}\\
           -\frac{H^\top D^{-1}\Phi_v}{2} & \Phi_a W & -\frac{\Phi_a W}{2}\\
           -\frac{\Phi_b\Delta^{-1}H^\top}{2}& -\frac{W\Phi_a}{2} & \Phi_b
        \end{pmatrix}
\end{equation}
\begin{equation*}
    P = \begin{pmatrix}
           O & D^{-1}H & O\\
           O & I-W & W\\
           \Delta^{-1}H^\top & O & O
    \end{pmatrix}
\end{equation*}
where $\Phi_v, \Phi_a, \Phi_b$ is  diagonal matrix consisting of $\phi(v) \,\mathrm{for}\, v\in V$, $\phi(v) \,\mathrm{for}\, v\in V_a$, $\phi(v) \,\mathrm{for}\, v\in V_b$.


\cite{Agarwal2006} demonstrated that traditional Laplacians for hypergraphs fall into the same analytical framework.  Based on the normalized Laplacian eigenvalue problem on star expansion, \cite{Agarwal2006} derived the node-related eigenvalue problem

\begin{equation}
    A_*A_*^\top x_v = (\lambda -1)^2 x_v
    \label{star eigenvalue problem}
\end{equation}
when denoting $A_* = D_v^{1/2}HW_*D_v^{1/2}$ with $w_*(e) = w(e)/\delta(e)$ and $\lambda$ is an eigenvalue of the normalized Laplacian. The same argument applies when $w_{*}(e)$ is set to 1, which we use for comparison.

By further normalizing our proposed Laplacian and considering a scenario where the edge weights in the expanded or reduced graph are identical, we can prove the following. If $W=I$, we have 
$$
\phi = \frac{1}{\mathrm{vol}(G)}(d_1, \ldots, d_n, \delta(1), \ldots, \delta(m), \delta(1), \ldots, \delta(m))
$$ as the unique stationary distribution vector:
\begin{align*}
    \phi P &= \frac{1}{\mathrm{vol}(G)}(\delta(1), \ldots, \delta(m))\\
    &= \frac{1}{\mathrm{vol}(G)}\begin{pmatrix}
        (\delta(1), \ldots, \delta(m))\Delta^{-1}H^\top \,&\, (d_1, \ldots, d_n)D^{-1}H \,&\, (\delta(1), \ldots, \delta(m))
    \end{pmatrix}\\
    &= \frac{1}{\mathrm{vol}(G)}\begin{pmatrix}
        \mathbf{1}^\top H^\top \,&\, \mathbf{1}^\top H \,&\, (\delta(1), \ldots, \delta(m))
    \end{pmatrix}\\
    &= \frac{1}{\mathrm{vol}(G)}(d_1, \ldots, d_n, \delta(1), \ldots, \delta(m), \delta(1), \ldots, \delta(m))
\end{align*}

This allows us to write

\begin{equation}
    \mathcal{L} = \Phi^{-1/2}L\Phi^{-1/2} = 
    \begin{pmatrix}
           I & -A/2 & -A/2\\
           -A^\top/2 & I & -I/2\\
           -A^\top/2 & -I/2 & I
    \end{pmatrix}
\end{equation}
with $A = D^{-1/2}H\Delta^{-1/2}$. In the similar way as \cite{Agarwal2006}, we can derive
\begin{equation}
    AA^\top x_v = (1-\lambda)(1-2\lambda) x_v
    \label{proposed eigenvalue problem}
\end{equation}

Thus, building on the analysis presented in \cite{Agarwal2006}, it can be shown that $A$ equals $A_*$ when the weight $W_*$ is artificially set to $I$, and that (\ref{star eigenvalue problem}) corresponds to the eigenvalue problem in (\ref{proposed eigenvalue problem}). However, this argument focuses on the normalized version of our Laplacian, assuming an artificial scenario in which all edges after expansion are identical. Future theoretical analysis may elucidate the connection between the two, but in the current paper we posit that this distinction lends our approach its unique character.


\subsection{Laplacian Anomaly Detection }




Because we aim to assess whether our approach better preserves spectral information in a dynamic setting, we use LAD~\cite{Huang2020} as our change point detection method. The basic idea behind LAD is to use the top $K$ singular values (or eigenvalues in our case) of the graph Laplacian at each time point as the embedding vector $v_{t}$ for the network. After this embedding vector is normalized,  the context matrix $C:=[v_{t-l},v_{t-l+1},\dots,v_{t-1}]$ is constructed by aggregating information from past network embedding vectors, where $l$ represents the window that determines the breadth of the lookback period. Using this context matrix, we can estimate the embedding vector $\hat{v_{t}}$ for time $t$ based on past behavior. We calculate the anomaly score by comparing the actual embedding vector $v_{t}$ at time $t$ with the predicted embedding vector from the context matrix. Specifically, we compute the anomaly score $Z$ as $Z = 1 - v_{t}^{\top} \hat{v}_{t}$ and the change point score as $\hat{Z}_{t} = \min(Z_{t} - Z_{t-1}, 0)$.










\section{Data}

\subsection{Synthetic Dataset with Ground-Truth Change Points}

We design the synthetic hypergraph dataset to model dynamic changes through four distinct change points. Initially, we divide 60 nodes into three clusters. The first change point involves reassigning some nodes within these clusters. At the second change point, we introduce a new cluster and reassign nodes from the existing clusters to this new cluster. At the third change point, we add complexity by introducing another new cluster and reassigning some nodes to it, along with adding 10 new nodes, all assigned to this new cluster. Finally, at the fourth change point, we remove a cluster and redistribute its nodes across the remaining clusters. We sample each hypergraph for each period using the code provided in \cite{Ruggeri2023}, with changes occurring every 30 time points. We set the block structure to be nearly diagonal, with small weights added to the off-diagonal parts. We sample 50 datasets for robust evaluation. We sample the hypergraph at multiple time intervals throughout these changes, providing a robust dataset for analyzing dynamic hypergraphs.

\subsection{Dynamic Legal Hypergraphs}

Court opinions are a rich source of information that blend factual evidence with legal principles to shape judicial decisions~\cite{Kondo2024}. Legal interpretation is a complex process in which court opinions and statutory laws are combined to construct the legal framework. In \cite{Hisano2024}, the authors showed that constructing a hypergraph, where each court judgment's laws and precedents form a single hyperedge, provides valuable insights into the legal structure. This approach employs nested degree-corrected stochastic block models \cite{Peixoto2014} applied to a proprietary Japanese dataset. Building on this work, we create a dynamic legal hypergraph using publicly available U.S. court opinion data.


\begin{figure}[h]%
\centering
\begin{subfigure}{0.48\textwidth}
    \centering
    \includegraphics[width=\textwidth]{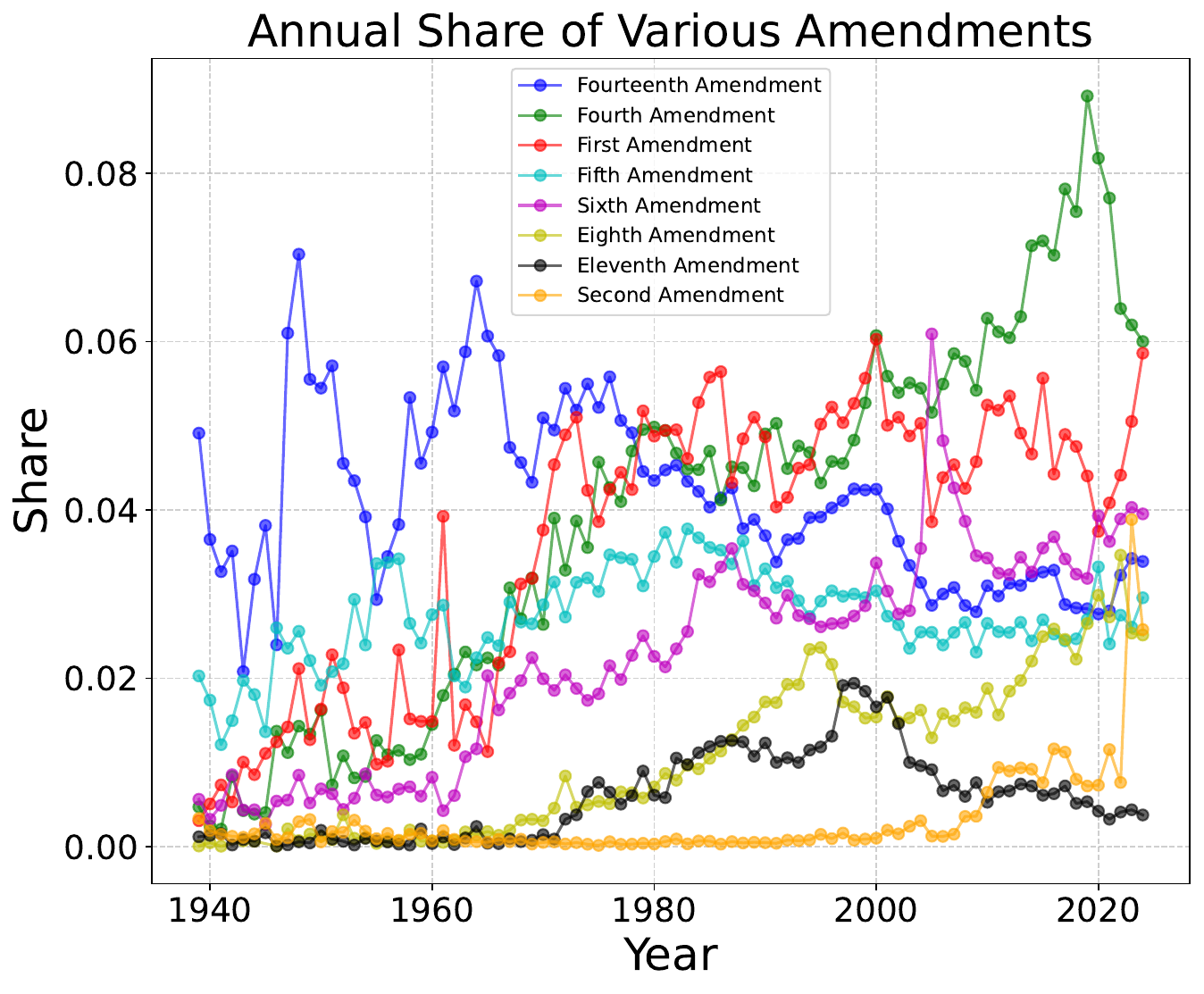}
    \caption{Time series plot depicting the share of the top eight amendments from 1939 to 2024.}
    \label{fig:ts-amend}
\end{subfigure}
\hfill
\begin{subfigure}{0.48\textwidth}
    \centering
    \includegraphics[width=\textwidth]{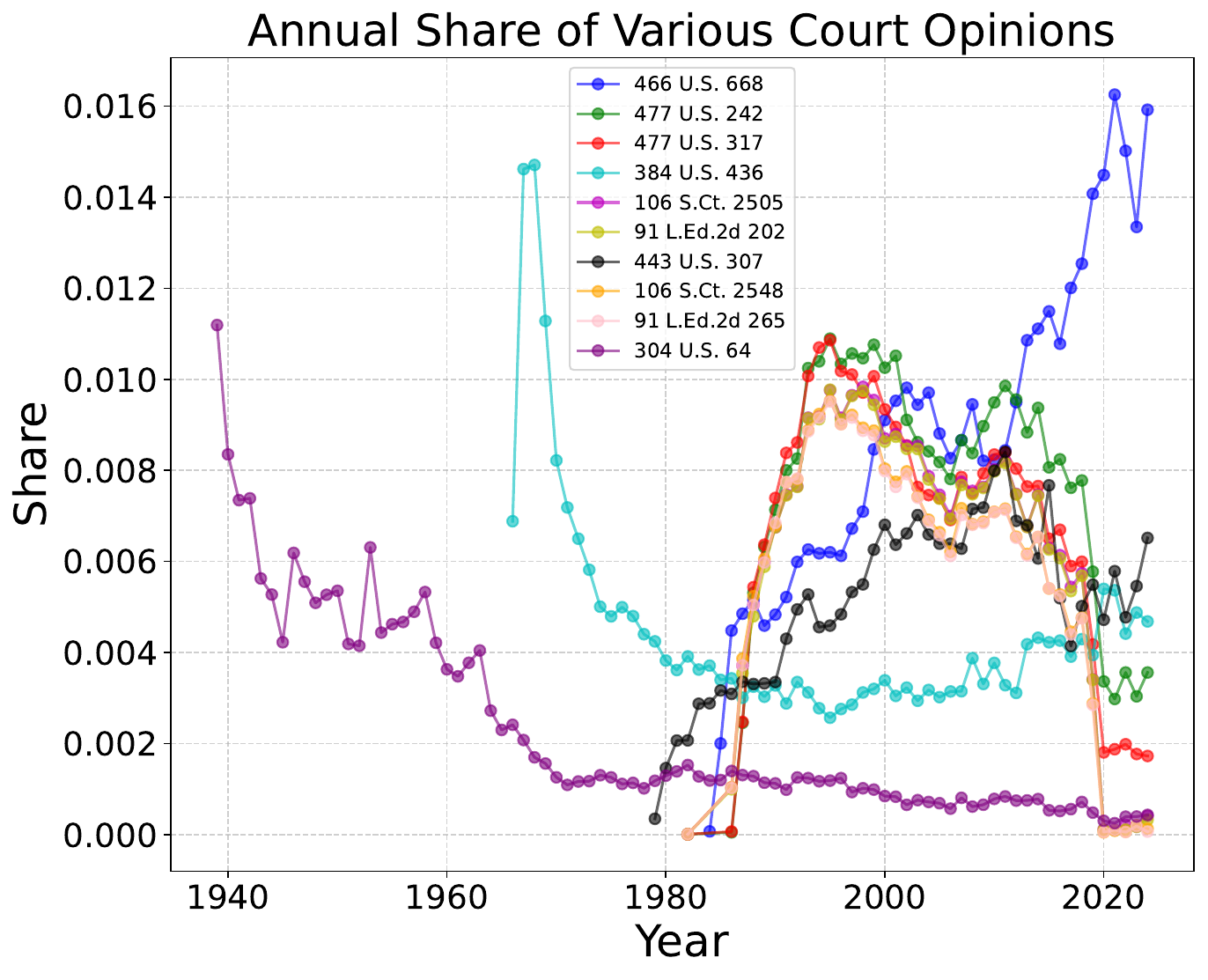}
    \caption{Time series plot depicting the share of the top 10 U.S. court opinions from 1939 to 2024.}
    \label{fig:ts-op}
\end{subfigure}
\caption{Time series plot depicting the share over time.}
\label{fig:ts}
\end{figure}

We obtained U.S. court opinion data from the Free Law Project \cite{FreeLawProject} using the dataset version dated May 6, 2024. We used this data to extract citation information and access opinion-cluster data to determine the publication dates of the court opinions. We extracted citation information using Eyecite, which is a tool provided by the Free Law Project \cite{Cushman2021}.


To manage the size of the dataset, we implemented a series of steps to refine the data for analysis. First, we concentrated on the top 1,000 most frequent amendments, legal statutes, and court opinions across multiple 10-year periods, ensuring comprehensive coverage of the entire dataset. We compiled these selected pairs into a list, ultimately narrowing down the data to 16,532 key items. Next, we identified the years for analysis by focusing on those that contained sufficient court opinions to generate meaningful hypergraphs. Specifically, we looked for periods of consecutive years where the count consistently exceeded 10,000. From these identified periods, we selected the most recent span for which legal activity remained above this threshold, which resulted in focusing on the years 1939--2024. In Fig.\ref{fig:ts}, we present the dynamic share of the top eight most-cited amendments and the top 10 most-cited U.S. court opinions from 1939 to 2024. These trends reveal noticeable dynamics that reflect the evolving legal structure of the U.S. court opinion data. All the synthetic datasets as well as the processed real datasets are available in our published code\footnote{\url{https://github.com/hisanor013/HypergraphCPD}}.

\section{Results}

\subsection{Synthetic Dataset}

First, we tested whether the adapted CB-gadget better preserved the spectral properties of hypergraphs compared with clique and star expansions. For change point detection, we used LAD \cite{Huang2020}, which identifies potential change points based on an anomaly score. We set a tolerance of ±2, considering a prediction correct if the true change point was within two time points. We evaluated performance using three metrics: the F1-score when predicting $3\%$ of the time points (four change points, matching the ground truth), the average F1-score for predictions using $3\%$--$15\%$ of all time points, and the average timing error when predicting $3\%$ of the time points. The timing error measures how close the predicted change points are to the actual change points by averaging the absolute differences between each true change point and its nearest prediction. The lower the timing error, the better the temporal accuracy. For the clique expansion, we used the full spectra. For the star and adapted CB-gadget expansions, we limited the analysis to the top 100 eigenvalues of the spectrum. We set the context window for LAD to 20.

\begin{table}[h]
\centering
\caption{Performance Comparison of Reduction Methods}
\label{table:reduction_methods}%
\begin{tabular}{@{}llll@{}}
\toprule
Reduction Method       & F1-Score (0.03) & Average F1 (0.03-0.15) & Timing Error\\
\midrule
Clique                 & 0.240           & 0.272    & 13.56  \\
Star                   & 0.675           & 0.450    & 8.01   \\
Adapted CB Gadget     & \textbf{0.750}  & \textbf{0.483}         & \textbf{6.07} \\
\bottomrule
\end{tabular}
\end{table}

\begin{figure}[h]%
\centering
\begin{subfigure}{0.33\textwidth}
    \centering
    \includegraphics[width=\textwidth]{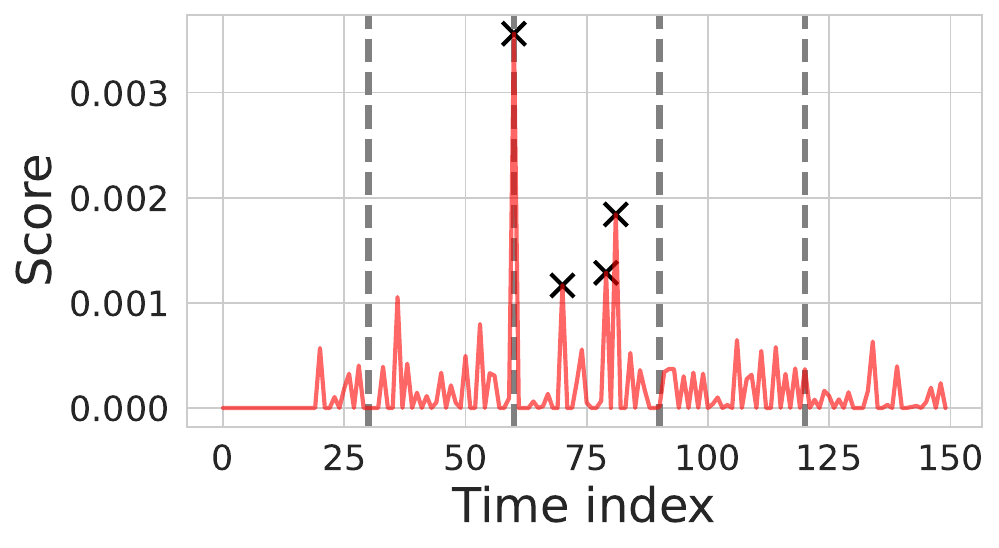}
    \caption{Clique}
    \label{fig:subfig1}
\end{subfigure}
\hfill
\begin{subfigure}{0.325\textwidth}
    \centering
    \includegraphics[width=\textwidth]{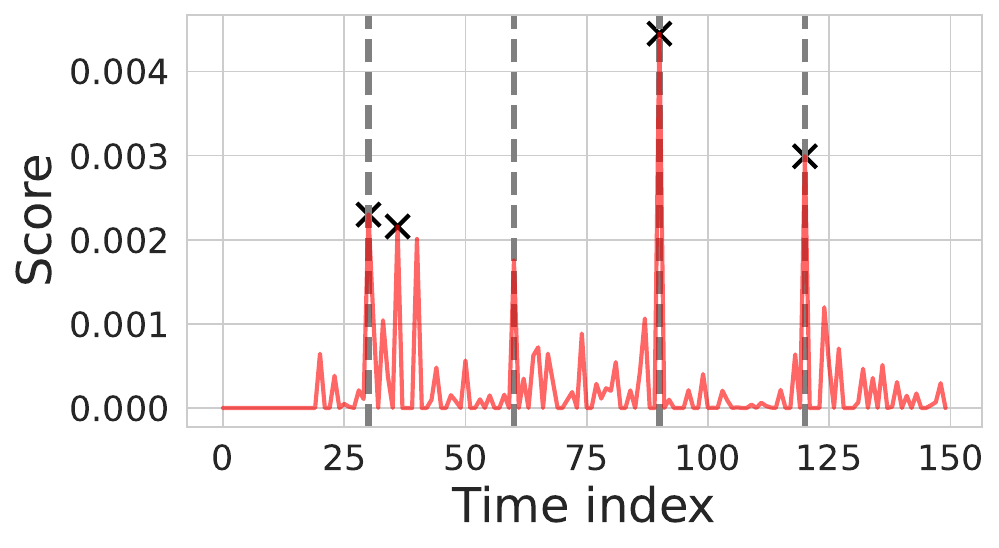}
    \caption{Star}
    \label{fig:subfig2}
\end{subfigure}
\hfill
\begin{subfigure}{0.325\textwidth}
    \centering
    \includegraphics[width=\textwidth]{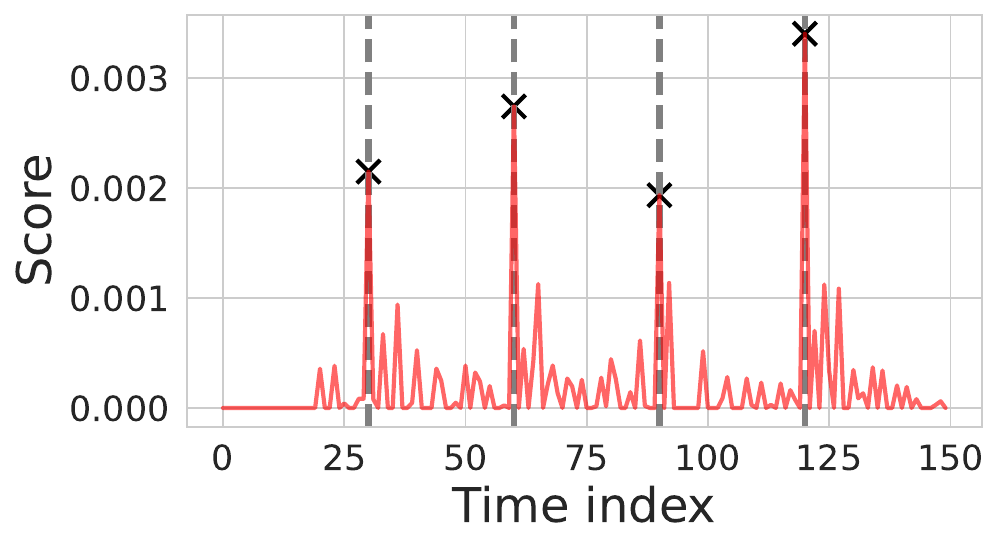}
    \caption{Adapted CB-gadget}
    \label{Adapted CB-gadget}
\end{subfigure}
\caption{Estimated performance of three reduction techniques on a synthetic dataset. The top $3\%$ of anomalous time points are marked by 'X,' and the vertical dashed line indicates the ground truth change points.}
\label{fig:synthetic-cp}
\end{figure}



Table~\ref{table:reduction_methods} presents the results, with each metric calculated as the average over 50 synthetic datasets. We observed that, across all evaluation metrics, the adapted CB-gadget outperformed both star and clique expansions. As demonstrated in the model selection process, the star expansion and our adapted CB-gadget share similar characteristics, which led to comparable performance. However, the added expressiveness of the adapted CB-gadget, which preserves the approximate conductance of the original hypergraph, resulted in superior performance. Additionally, we include figures that show the performance of each approach on a single synthetic dataset for reference (Fig.\ref{fig:synthetic-cp}).

\subsection{Dynamic Legal Hypergraph}

For the real dataset, we configured LAD to identify the top 0.05 dates (i.e., four dates) with the highest change point scores. These results are displayed in Fig. \ref{fig:us-cp}. The algorithm pinpointed change points in the years 1965, 1967, 1988, and 2020. The years 1965 and 1967 align with the Warren Court era, marked by the landmark decision in Miranda v. Arizona, 384 U.S. 436 (1966). This case revolutionized the criminal code and its impact  is evident: the citation count surged in 1966, which made the ruling one of the most frequently cited to this day.

The change point detected in 1988 can be attributed to significant shifts in legal standards established by key Supreme Court decisions in the mid-1980s. Notably, Strickland v. Washington, 466 U.S. 668 (1984), set a precedent ensuring that court proceedings comply with the Sixth Amendment's requirement for a fair trial. Subsequent decisions, Anderson v. Liberty Lobby, Inc., 477 U.S. 242 (1986) and Celotex Corp. v. Catrett, 477 U.S. 317 (1986), part of the 1986 summary judgment trilogy \cite{Yamamoto1990}, refined the criteria for summary judgment, emphasizing the sufficiency of evidence required for a case to proceed to trial. These rulings collectively reshaped the legal landscape by influencing how courts evaluate procedural fairness, thereby fostering more modern and standardized approaches in legal adjudication.


The 2020 change point presents a unique challenge. The pandemic's escalation led to an increase in surveillance technologies to track the virus, raising significant privacy concerns. This, in turn, is likely to have spurred discussions and legal challenges concerning the balance between public health needs and individual privacy rights, thereby leading to an increased citation of the Fourth Amendment. As illustrated in the figure, the prominence of the Fourth Amendment peaked in 2020, which may be attributed to these developments.

\begin{figure}[h]%
\centering
\includegraphics[width=0.7\textwidth]{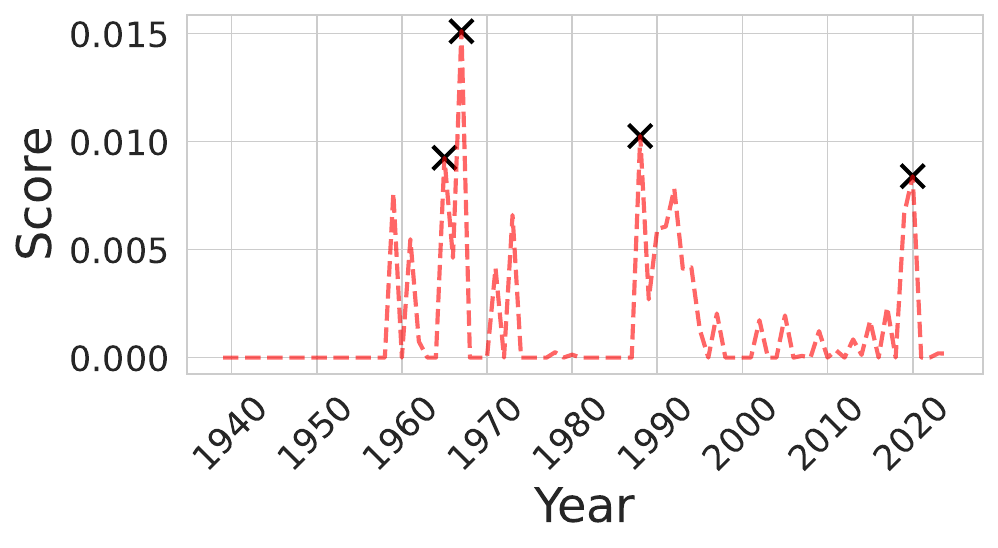}
\caption{Results for the real data. The top $5\%$ of anomalous time points are marked by 'X,'}
\label{fig:us-cp}
\end{figure}

\section{Conclusion}

In this paper, we introduced a novel method for transforming hypergraphs into strongly connected weighted directed graphs using an adapted CB-gadget and a symmetrized combinatorial Laplacian. We proved that the conductance of the transformed graph upper bounded the harmonic mean of the conductance and edge expansion of the original hypergraph, thereby preserving essential cut information. Additionally, We also provided analysis showing how our Laplacian relates to that derived from the star expansion. Experiments on both synthetic and real datasets, including a dynamic legal hypergraph from U.S. court opinions, demonstrated that our method outperformed existing clique and star expansions in preserving spectral information required for change point detection.

\section{Acknowledgment}
This research was supported by JST FOREST Program (Grant Number JPMJFR216Q), the UTEC-UTokyo FSI Research Grant Program, and the Grant-in-Aid for Scientific Research (KAKENHI) (JP24K03043). We thank Edanz for editing a draft of this manuscript.

\bibliographystyle{spmpsci} 
\bibliography{refs} 

\end{document}